\documentclass[11pt]{article}

\usepackage[margin=1in]{geometry}
\usepackage{amsmath,amssymb,amsthm}
\usepackage{hyperref}
\usepackage{xspace}

\newcommand{\NN}{\mathbb{N}}

\newcommand{\JJ}{\mathcal{J}} 
\newcommand{\II}{\mathcal{I}} 
\newcommand{\OO}{\mathcal{O}} 
\NewDocumentCommand{\bell}{O{\JJ} O{\II} O{\OO}}{B(#1,#2,#3)} 
\newcommand{\xx}{\mathbf{x}} 
\newcommand{\yy}{\mathbf{y}} 
\newcommand{\zz}{\mathbf{z}} 
\newcommand{\DD}{\mathcal{D}} 

\newcommand{\ie}{\text{i.e.}\xspace}
\newcommand{\eg}{\text{e.g.}\xspace}

\numberwithin{equation}{section}
\newtheorem{theorem}{Theorem}[section]
\newtheorem{definition}[theorem]{Definition}

\begin{document}

\title{Comment on `The axiom of choice and the no-signalling principle'}

\author{Martti Karvonen\\
Department of Computer Science, University College London\\
\texttt{martti.karvonen@ucl.ac.uk}}
\date{}

\maketitle

\begin{abstract}
The main claim of Baumeler et al. (Proc. R. Soc. A: Math. Phys. Eng. Sci. 481, 20240601) is that  ``functional (deterministic) no-signalling resources can be stronger than probabilistic ones'' in a certain nonlocal game on a Bell scenario with countably many parties.  
We disagree and argue that (i) under standard definitions,  deterministic no-signalling resources are always probabilistic no-signalling resources;  (ii) the deterministic strategy considered in Baumeler et al. can be promoted to
a genuinely probabilistic strategy with similar properties; and (iii) a key step in a derivation in Baumeler et al., claimed to hold for all no-signalling strategies, implicitly assumes measurability, leaving a gap in the argument. We propose measurability assumptions which we conjecture would make this derivation rigorous. Taken together, the phenomenon highlighted by Baumeler et al. is best understood as a difference between measurable and non-measurable no-signalling resources. 
 \end{abstract}

\section{A brief summary of the study by Baumeler et al.}

Baumeler et al.~\cite{ACandNS} consider a Bell scenario with countably many parties. Modelled on a certain hat puzzle~\cite{hatproblems}, a referee first chooses one of two hat colours for each party, independently and uniformly at random, resulting in an infinite binary sequence. Each party then receives as their input the suffix listing the colours of all later parties and must guess their own hat colour without any communication with the others.  Intuitively, as the colours are picked independently, one gains no information from the colours of other hats and thus cannot do any better than a random guess, winning half of the time.  Building on this intuition, Baumeler et al.~\cite{ACandNS} claim a proof that the ratio of players that win converges to one half almost surely, and that this holds for all no-signalling strategies. However, using the axiom of choice, Baumeler et al.~\cite{ACandNS} construct a deterministic no-signalling strategy which guarantees, for every input, that all but finitely many of the players win.

This is interpreted in~\cite{ACandNS} as showcasing a difference between deterministic and probabilistic no-signalling. In particular, Baumeler et al.~\cite{ACandNS} claim that some deterministic no-signalling strategies are not probabilistic no-signalling strategies.  We disagree and instead believe that the phenomenon highlighted by~\cite{ACandNS} is best understood as a difference between measurable and non-measurable strategies. 

We support this viewpoint in three ways. First, we prove that, under standard definitions, also agreed upon by the authors of~\cite{ACandNS}, deterministic no-signalling resources are always probabilistic no-signalling resources (albeit of a trivial kind). In particular, the deterministic strategy constructed in~\cite{ACandNS} satisfies the probabilistic no-signalling principle as usually understood. Second, this deterministic strategy can be used to construct a genuinely probabilistic no-signalling strategy with similarly paradoxical behaviour. Third, we show that the proof in~\cite{ACandNS} that the ratio of winning players converges to one half implicitly assumes measurability when applying Azuma’s theorem~\cite{Azuma1967}, and we propose measurability assumptions that we conjecture would make the argument rigorous.

\section{Deterministic no-signalling implies probabilistic no-signalling}
In this section we demonstrate that all deterministic no-signalling resources are (degenerate) probabilistic resources, ruling out the separation claimed in~\cite{ACandNS} under standard definitions. We explain this in a self-contained manner, assuming only a modicum of measure theory (\eg~\cite{Billingsley:probabilityandmeasure}).  A Bell scenario formalizes the idea of a set of parties performing experiments on a large shared system, with each party choosing their own measurement setting for their part of it and obtaining their own outcome. 

\begin{definition} A  \emph{Bell scenario} $\bell$ is specified by fixing (i) a (possibly infinite) set $\JJ$ of parties, (ii) a (possibly infinite) set $\II$ of inputs, (iii) a measurable space $\OO$ of outcomes. The \emph{joint measurements} of $\bell$ are given by tuples $\xx=(x_j)_{j\in \JJ}$ in  $\II^\JJ$ fixing a measurement in $\II$ for each party in $\JJ$, and \emph{joint outcomes} are elements of the measurable space $\OO^\JJ$.
\end{definition}

While our outcome set will be a finite and discrete measurable space, we still need to invoke measure theory to define joint probability distributions over joint measurements, since they live in the product $\OO^\JJ=\prod_{j\in \JJ} \OO$. We equip $\OO^\JJ$ with the standard product $\sigma$-algebra generated by the cylinder sets $\prod_{j\in\JJ} X_j$ where each $X_j\subset \OO$ is measurable and $X_j=\OO$ for all but finitely many $j\in\JJ$. This ensures that for each $U\subset \JJ$, the projection $\pi_U\colon\prod_{j\in \JJ} \OO\to \prod_{j\in U}\OO$  onto $U$ is  measurable. When viewing elements of the product as functions, we write $\xx|_U\in \OO^U$ for the restriction of $\xx\in \OO^\JJ$ to $U$.

An empirical model\footnote{Also known as a  correlation, behaviour, or a strategy, depending on the source.} is an idealization of the results of such an experiment, whether as averages over multiple runs or as theoretically predicted probabilities.

\begin{definition} An empirical model over $\bell$ consists of a family $e=(e_\xx)_{\xx\in \II^\JJ}$ where for each joint input $\xx\in \II^\JJ$ we have a probability measure
$e_\xx$  on the space $\OO^\JJ$ of joint outcomes. We write $e:\bell$ to denote that $e$ is an empirical model over $\bell$.
\end{definition}

Arbitrary empirical models are too general in certain physically motivated situations. If the parties are space-like separated, then choices made by some parties cannot causally influence the outcomes of others,  assuming relativity. It is then natural to postulate the  \emph{no-signalling conditions} which assert that the distribution for a given set of parties does not depend on the inputs of the parties in the complement. Formally, we note that each restriction $|_U\colon  \OO^\JJ\to \OO^U$ is measurable. As a result, for any probability measure $d$ on the space  $\OO^\JJ$ of joint outcomes, there is a well-defined push-forward probability measure on $\OO^U$ along the restriction to $U$ which is usually known as the marginalization of $d$ to $U$. We overload notation and denote the resulting probability measure by $d|_U$. With this out of the way, we can formally define no-signalling.

\begin{definition}
An empirical model $e:\bell$ is a \emph{no-signalling empirical model}, if for all joint inputs $\xx,\yy\in \II^\JJ$ we have
\begin{equation}\label{eq:NS}
e_\xx|_{[\xx=\yy]}=e_\yy|_{[\xx=\yy]}
\end{equation}
where $[\xx=\yy]$ denotes the set $\{j\in \JJ:\xx_j=\yy_j\}$ of indices where $\xx$ and $\yy$ agree.
\end{definition}

We  emphasize that these definitions of (no-signalling) empirical models are standard: for instance, they agree with those given in~\cite{CVContextuality} which directly generalize the well-established Abramsky--Brandenburger framework for contextuality~\cite{ab} into the infinitary setting. Moreover, in an email exchange about these issues when~\cite{ACandNS} first appeared on arXiv, the corresponding author confirmed that  ``this is the natural definition of (probabilistic) no-signaling and it is indeed the one that we adopted throughout''~\cite{the_email}.

That said, one may well wonder why we do not require that the distributions $e_\xx$ depend measurably on $\xx$, forming a \emph{Markov kernel}. To our awareness, this requirement has not been explicitly discussed in prior literature on Bell scenarios and contextuality, mainly because this assumption holds automatically when there are finitely many parties and inputs (so that the space of joint inputs is a finite, discrete measurable space), and the case of infinitely many measurements, even when formally allowed as in~\cite{CVContextuality}, has rarely been studied in detail. As we believe that the phenomenon highlighted in~\cite{ACandNS} is best understood as displaying a difference between measurable and non-measurable empirical models, we could view~\cite{ACandNS} as giving an example of the kind of pathology that can be avoided by requiring empirical models to be measurable in this sense. We return to this measurability assumption in the conclusion and ultimately advocate for it. 

We stress that, due to the measurability of the restrictions, the probability measures appearing in \eqref{eq:NS} are always well defined, no matter how pathological $e$ is otherwise. This is not to say that every empirical model is always no-signalling, only that the question of whether it is no-signalling is always well-posed.  We now define the class of local empirical models.

\begin{definition} Let $\mu$ be a probability measure on  $\OO^{\JJ\times\II}$. Define an empirical model $e_\mu=(e_{\mu,\xx})_{\xx\in \II^\JJ}$ by setting $(e_{\mu,\xx}):=\mu|_\xx$, where we write $\mu|_\xx$ for the push-forward of $\mu$ along the restriction to the graph $\{(j,\xx_j): j\in\JJ\}\subset \JJ\times \II$ of $\xx\in \II^\JJ$. We say that an empirical model $e$ is \emph{local} if $e=e_\mu$ for some probability measure $\mu$ on  $\OO^{\JJ\times\II}$.
\end{definition}

It is well known that quantum mechanics predicts~\cite{bell64,nonlocalityreview} and experiments confirm~\cite{hensen2015loophole} the existence of nonlocal empirical models. 

\begin{theorem}\label{thm:NCisNS} All local empirical models are no-signalling. 
\end{theorem}
\begin{proof} Follows directly from functoriality of taking push-forward measures, \ie the fact that pushing a measure forward twice is the same as pushing it forward once along the composite measurable mapping. Explicitly, let $e_\mu$ be local, $\xx,\yy\in \II^\JJ$ be arbitrary, and compute as follows:
\[e_{\mu,\xx}|_{[\xx=\yy]}=(\mu|_\xx)|_{[\xx=\yy]}=\mu|_{[\xx=\yy]}=(\mu|_\yy)|_{[\xx=\yy]}=e_{\mu,\yy}|_{[\xx=\yy]}.\qedhere
\]
\end{proof}
In particular, if $\mu$ is given by the Dirac measure $\delta(f)$ on some function $f\colon \JJ\times\II\to \OO$, then it always induces a no-signalling empirical model. Similarly, finite and countable convex combinations of such models are no-signalling, \ie if $\mu_i$ for $i\in\NN$ are Dirac  measures and $r_i$ for $i\in\NN$ are nonnegative reals summing to $1$, then $\mu:=\sum_{i\in\NN} r_i \mu_i$ induces a no-signalling empirical model. 

We now turn our attention to the \emph{functional no-signalling principle} discussed in~\cite[Equation 1.3]{ACandNS}. In this setting, one considers a function $f\colon \II^\JJ\to \OO^\JJ$ from joint inputs to joint outputs, written as $f(\xx)=(f_j(\xx))_{j\in \JJ}$ for $\xx=(x_j)_{j\in \JJ}$ in $\II^\JJ$. Any such function $f$ induces an empirical model $e(f): \bell$ by setting the distribution $e(f)|_\xx$ on $\OO^\JJ$ to be the Dirac measure on $f(\xx)$.

Such a function is said to satisfy the \emph{functional no-signalling principle} if the $j$th output depends only on the $j$th input, which can be expressed by saying that for all $j\in\JJ$ and all $\xx,\yy\in \II^\JJ$, if $x_j=y_j$, then $f_j(\xx)=f_j(\yy)$. In that case the function $f\colon \II^\JJ\to \OO^\JJ$  in fact determines a function $\hat{f}\colon\JJ\times\II\to \OO$, defined by $\hat{f}(j,i)=f_j(\xx)$, where $\xx$ is an arbitrary element of $\II^\JJ$ satisfying $x_j=i$. As $f$ satisfies the functional no-signalling principle, $\hat{f}$ is indeed a well-defined function, as the choice of $\xx$ does not affect the end result. Moreover, $e(f)=e_{\delta(\hat{f})}$ \ie  $e(f)$ is induced by the Dirac measure on $\hat{f}$, so that $e(f)$ is local. Therefore, Theorem~\ref{thm:NCisNS} immediately gives the following result.
\begin{theorem} If  $f\colon \II^\JJ\to \OO^\JJ$ satisfies the functional no-signalling principle, it induces a no-signalling empirical model $e(f)$ on $\bell$.
\end{theorem}
Thus all functional no-signalling resources are no-signalling probabilistic resources (albeit of a degenerate kind), ruling out the claimed separation under standard definitions. 

\section{A probabilistic counterexample}

We now explain the setting of~\cite{ACandNS} in more detail. The main construction is based on an `infinite hat puzzle'~\cite{hatproblems}: there is a player for each natural number, and a referee equips each player with a hat with one of two colours $\{0,1\}$. The $j$th player sees the colours of the hats of all later players (\ie the $i$th hat for $i>j$) but no other hats. The players cannot communicate after the setup but may agree on a strategy beforehand, possibly using shared randomness, and each must guess the colour of their own hat based on what they see and the agreed strategy, with the referee recording for each player whether they won (by guessing correctly) or lost.

We model the game on the Bell scenario  $\bell[\NN][\{0,1\}^\NN][\{0,1\}]$ with parties given by $\NN$, inputs\footnote{Baumeler et al.~\cite{ACandNS} work with inputs given by the unit interval $I=[0,1]$. Taking binary expansions maps the unit interval to $\{0,1\}^\NN$, and yields an isomorphism of probability spaces, up to null sets corresponding to numbers with two binary expansions and the corresponding binary sequences, when the unit interval is endowed with the uniform distribution (Lebesgue measure) and  $\{0,1\}^\NN$ with the (completion of the) product of infinitely many coin flips. Therefore, we can view this choice as just a matter of notational convenience. Moreover, working with the slightly less familiar measurable space  $\{0,1\}^\NN$ discourages computing integrals without first verifying that they are well defined.} by $\{0,1\}^\NN$ \ie sequences of colours, and two outputs $\{0,1\}$. The game itself can be modelled as follows: the referee first picks each hat independently and uniformly at random, \ie samples a sequence $\xx=(x_0,x_1,x_2,\dots)$  from the uniform product distribution on $\{0,1\}^\NN$. The interpretation is that $x_j\in\{0,1\}$ is the colour of the $j$th hat. The referee then sends each player an input: the $j$th player is given the input $(x_{j+1},x_{j+2},\dots)$, \ie the colours of the hats of all the later players. They then output their guess $a_j$, and the result $r_j$ of the $j$th player is then given by
\[r_j=\begin{cases} 1 \text{ if }a_j=x_j \\
					0 \text{ otherwise.}
	  \end{cases}\]
Assuming the players have no limits on their memory or computational abilities, they can counter-intuitively construct a strategy that guarantees that all but finitely many of them win, using the \emph{axiom of choice}. Consider the equivalence relation $\sim$ on $\{0,1\}^\NN$ of having equal tails after some threshold: formally, $\xx\sim\yy$ if there is an $n\in\NN$ such that $x_j=y_j$ for all $j>n$. Using the axiom of choice, the players fix one representative in each equivalence class and encode these choices in a function $g\colon \{0,1\}^\NN\to\{0,1\}^\NN$ sending each $\xx$ to the chosen representative of its class. The strategy of the $j$th player is then as follows: given the input $(x_{j+1},x_{j+2},\dots)$, they prepend it with $j+1$ zeroes, obtaining a sequence $\yy_j:=(0,\dots , 0 ,x_{j+1},x_{j+2},\dots)\in \{0,1\}^\NN$, and set their guess $a_j$ to be $g(\yy_j)_j$. By construction, for each $j\in\NN$ we have $\yy_j\sim \xx$ and therefore $g(\yy_j)=g(\xx)$. As $g(\xx)\sim\xx$, this means that after some threshold $n(\xx)\in\NN$, all players guess correctly.  

This strategy has each player deterministically choose their outcome based only on their own input, thus inducing a function $\NN\times \{0,1\}^\NN\to\{0,1\}$ and hence by Theorem~\ref{thm:NCisNS} a no-signalling empirical model on $\bell[\NN][\{0,1\}^\NN][\{0,1\}]$.  Now, Baumeler et al.~\cite{ACandNS} summarize their main result as follows: ``Hence, imposing probabilistic NS, the ratio of players that win the game over the total converges to 1/2 almost surely.'' However, the aforementioned strategy gives a counterexample. This counterexample is in fact discussed in~\cite{ACandNS} but interpreted differently: it is treated as a functional no-signalling strategy which fails to be a probabilistic no-signalling model. As far as we can understand, their argument in~\cite[Section 1(a)]{ACandNS} for this is as follows: the no-signalling equations~\eqref{eq:NS}, when formulated rigorously using measure theory, require the (joint) inputs and outputs to  ``be the values of random variables'' which formally are measurable functions on some sample space. However, the axiom of choice can produce non-measurable sets (such as the set of chosen representatives above) giving rise to non-measurable functions which in turn fail to be random variables, so that~\eqref{eq:NS} may not be well-posed when the empirical model is constructed using the axiom of choice. 

To put our objection bluntly, just because some functions associated to an empirical model may fail to be measurable (such as the function $\xx\mapsto e_\xx$ sending a joint input to a joint distribution of outcomes), it does not mean that all of them are non-measurable. In particular, the no-signalling equations always involve honest probability distributions due to the measurability of the restrictions $|_U\colon  \OO^\JJ\to \OO^U$, and the empirical model constructed in~\cite{ACandNS} discussed earlier satisfies the no-signalling equations in the usual, probabilistic sense.

Alternatively,~\cite{ACandNS} could be seen as implicitly assuming that empirical models satisfying~\eqref{eq:NS} have to be measurable. In this view, our contribution is (i) making this assumption explicit, (ii) observing that this assumption is not necessary for~\eqref{eq:NS}, and (iii) pointing out that this is an assumption of \emph{measurability} rather than a difference between probabilistic as opposed to deterministic strategies. To emphasize that this is not restricted to deterministic strategies, let $g_j$ be obtained from $g$ by flipping the first $j$ outputs and keeping the rest as in $g$. Let $e_j$ be the empirical model induced by $g_j$, and define $d:=\frac{1}{2}e_0+\frac{1}{2}e_1$ and $e:=\sum_{j\in\NN}\frac{1}{2^{j+1}}e_j$. These models are no-signalling by Theorem~\ref{thm:NCisNS}: the behaviour of $d$ is uniformly random at the first party (and deterministic otherwise), while the behaviour of $e$ is genuinely probabilistic for every party (albeit it tends to deterministic behaviour exponentially fast). Moreover, for each fixed joint input, the strategy  $d$ guarantees that all but finitely many of the players win, and for $e$ this happens almost surely (by an application of the first Borel--Cantelli lemma). It follows that neither satisfies the bounds derived in~\cite{ACandNS}. These two genuinely probabilistic models show that the findings of~\cite{ACandNS} do not demonstrate a separation between deterministic and probabilistic no‑signalling strategies under standard definitions.

\section{Missing measurability assumption} 

In~\cite{ACandNS}, it is claimed that, under probabilistic no‑signalling, the ratio of winning players converges to one half almost surely.  We have now exhibited a probabilistic counterexample, so this claim does not hold for all no-signalling probabilistic empirical models. We next identify a gap in their derivation, and propose a stronger assumption which we conjecture to lead to a corrected proof, applying only to a subclass of no-signalling empirical models. 

The key gap we have identified is the following: at a crucial juncture the argument applies Azuma's theorem~\cite{Azuma1967} which concerns certain sums of random variables. Here, the candidate random variables $S_n$ are given by sums $S_n=\sum_{j=1}^n s_j$ of ``success random variables'' $s_j$ where $s_j=1$ if the $j$th player wins and $s_j=-1$ otherwise. 

Recall that, formally, a real-valued random variable is a measurable function into the reals. The domain of this function $s_j$ is never explicitly defined in~\cite{ACandNS}, but it is reasonable to take it to be the space  $\{0,1\}^\NN$, since that is the sample space for the referee’s random assignments of hats (or inputs).

We now argue that these ``random variables'' need not be measurable functions, in which case the assumptions of Azuma's theorem are not satisfied. The functions $S_n$ are measurable iff the functions $s_j$ are, so we focus on the latter. We first consider the situation when players are using the deterministic strategy constructed earlier using the function $g\colon \{0,1\}^\NN\to \{0,1\}^\NN$ obtained via the axiom of choice. Then $s_j$ is given by
\[s_j(\xx):=\begin{cases} 1 \text{ if }g(\xx)_j=x_j \\
					-1 \text{ otherwise.}
	  \end{cases}\]
As this function takes only two values, we see that $s_j$ is measurable iff the set $G_j:=\{\xx\in\{0,1\}^\NN:g(\xx)_j=x_j\}$ of inputs for which the $j$th party wins when following the strategy induced by $g$ is measurable. If the sets $G_j$ were measurable, so would be the set 
$G:=\bigcap_{j=1}^\infty G_j=\{\xx\in\{0,1\}^\NN:g(\xx)=\xx\}$ of inputs where all players win. This coincides with the set of chosen representatives, which is a direct analogue of the well-known Vitali set on the reals (\eg~\cite[Section 3]{Billingsley:probabilityandmeasure}) and can be shown to be non-measurable similarly.  

For the sake of completeness, we now prove that the set $G$ of chosen representatives is not measurable. Recall first that pointwise addition modulo $2$ makes $\{0,1\}^\NN$ into a group, and using this group operation we can translate a subset $A\subset \{0,1\}^\NN$ using an element $\xx\in\{0,1\}^\NN$, resulting in the translated set $\xx+A:=\{\xx+\yy:\yy\in A\}$. Write $F$ for the subset of $\{0,1\}^\NN$ consisting of all finitely supported infinite sequences, \ie elements $\xx$ of  $\{0,1\}^\NN$ which are zero on all but finitely many entries. Note that $F$ is countable (for example, it is a countable union over the finite sets $F_i$ of sequences which are zero after $i$). 

We now observe that for distinct elements $\xx$ and $\yy$ of $F$ the sets $\xx+G$ and $\yy+G$ are disjoint. If not, pick some $\zz\in \xx+G\cap \yy+G$, which can be written as $\zz=\xx+r_1=\yy+r_2$ for some $r_1,r_2\in G$. As $\xx$ and $\yy$ are finitely supported, we have $r_1\sim r_2$ and therefore $r_1=r_2$ as $G$ has exactly one element from each equivalence class. Therefore $\zz=\xx+r_1=\yy+r_2$ implies $\xx=\yy$, a contradiction.  Moreover, $\bigcup_{\xx\in F} \xx+G=\{0,1\}^\NN$ as for any $\yy\in \{0,1\}^\NN$ the sequences $\yy$ and $g(\yy)$ differ at finitely many points, so that $\xx:=\yy+g(\yy)\in F$ satisfies $\yy=\xx+g(\yy)$ and therefore $\yy\in \xx+G$. We have thus expressed $\{0,1\}^\NN$ as a countable union of disjoint sets (the translates of $G$).  

Now, let $\mu$ be the measure of uniform independent coinflips. This measure is in fact the Haar measure on the (compact) group $\{0,1\}^\NN$, which in particular means that it is translation invariant, \ie $\mu(A)=\mu(\xx+A)$ for any $\xx\in \{0,1\}^\NN$ and any measurable $A$\footnote{Translating along $\xx$ gives a measurable bijection on $\{0,1\}^\NN$, so it preserves measurability.}. Now if $G$ were measurable, 
\[1=\mu(\{0,1\}^\NN)=\mu\biggl(\bigcup_{\xx\in F} \xx+G\biggr)=\sum_{\xx\in F}\mu (\xx+G)=\sum_{\xx\in F}\mu (G)=\sum_{i\in\NN}\mu (G)\] 
where the last step is by countability of $F$. However, if $\mu(G)=0$, then $\sum_{i\in\NN}\mu (G)=0$, and otherwise this sum diverges, so there is no possible value for $\mu(G)$ in the unit interval satisfying this. Thus we reach a contradiction, showing that $G$ is in fact not measurable. 

Although the point is clearest for deterministic strategies, the mixed model $d:=\frac{1}{2}e_0+\frac{1}{2}e_1$  above exhibits the same issues for all $j\geq 1$. For $e:=\sum_{j\in\NN}\frac{1}{2^{j+1}}e_j$, one can also show that the Markov kernels $s_j\colon \{0,1\}^\NN\to \{-1,1\}$ are not measurable, although this takes some more work. Interestingly, none of these examples are counterexamples to the derived bound in the sense of violating it: rather, they fall outside its scope, as the probabilities the bound concerns do not exist for these strategies.

\section{Conclusion}
The main phenomenon highlighted in~\cite{ACandNS} is not about deterministic versus probabilistic strategies, nor about the no-signalling principle itself. Rather, the issue is that an empirical model may satisfy the no-signalling principle yet be pathological in the sense of being non-measurable.

To formalize this, recall that an empirical model $e: \bell$ is defined as a family $e=(e_\xx)_{\xx\in \II^\JJ}$. Equivalently, it can be thought of as a function sending each joint input $\xx\in \II^\JJ$ to the distribution $e_\xx$ over joint outcomes. It is known that the set $\DD (X)$ of distributions on a measurable space $X$ can itself be made into a measurable space, so that measurable maps $X\to\DD (Y)$ correspond to stochastic maps (Markov kernels) $X\to Y$~\cite{Giry1982}. We propose to call an empirical model $e: \bell$  \emph{measurable} if it is measurable when viewed as a function from joint inputs to distributions over outcomes with respect to the product $\sigma$-algebras fixed above.  We conjecture that the results of~\cite{ACandNS} apply to all measurable empirical models, and hence the paper contributes another example to the long list of differences between measurable and non-measurable behaviours. Moreover, we believe that it is reasonable to require that all empirical models be measurable in this sense. The main reason is operational: one can compose any measurable empirical model with any distribution on joint inputs and any measurable post-processing on outputs and obtain a well-defined probability measure. In particular, for any nonlocal game with a measurable win predicate, the probability of winning under any input distribution is well-defined. The assumption also rules out the pathological examples discussed earlier, and holds trivially when the sets of parties and inputs are finite. In any case, it is crucial in science to state one’s assumptions explicitly.

\paragraph{Acknowledgements.} Martti Karvonen acknowledges useful conversations with Rui Soares Barbosa and support from the EPSRC grant EP/V040944/1 Resources in Computation.

\bibliographystyle{unsrt} 

\bibliography{refs}

\end{document}